\documentclass[11pt]{article}
\usepackage[margin=1in]{geometry}
\usepackage{amssymb,amsfonts,amsmath,amsthm,amscd,dsfont,mathrsfs,bbold,pifont}
\usepackage{blkarray}
\usepackage{graphicx,float,psfrag,epsfig,color,yhmath}
\usepackage{comment}
\usepackage{subcaption}

\usepackage{algorithm}
\usepackage[noend]{algpseudocode}

\makeatletter
\def\BState{\State\hskip-\ALG@thistlm}
\makeatother

	\usepackage{microtype}
	\usepackage[pdftex,pagebackref=true,colorlinks]{hyperref}
	\hypersetup{linkcolor=[rgb]{.7,0,0}}
	\hypersetup{citecolor=[rgb]{0,.7,0}}
	\hypersetup{urlcolor=[rgb]{.7,0,.7}}

\newtheorem{theorem}{Theorem}

\newtheorem{lemma}{Lemma}
\newtheorem{corollary}{Corollary}
\newtheorem{definition}{Definition}

\begin{document}
\title{Reed-Muller codes have vanishing bit-error probability below capacity: a simple tighter proof via camellia boosting}
\author{Emmanuel Abbe and Colin Sandon \\ EPFL}
\date{}
\maketitle

\begin{abstract}
This paper shows that a class of codes such as Reed-Muller (RM) codes have vanishing bit-error probability below capacity on  symmetric channels. The proof relies on the notion of   `camellia codes': a class of symmetric codes decomposable into  `camellias', i.e., set systems that differ from sunflowers by allowing for scattered petal overlaps. The proof then follows from a boosting argument on the camellia petals with second moment Fourier analysis. For erasure channels, this gives a self-contained proof of the bit-error result in \cite{Kudekar17}, without relying on sharp thresholds for monotone properties \cite{friedgut1996every}. For error channels, this gives a shortened proof of \cite{reeves} with an exponentially tighter bound, and a proof variant of the bit-error result in \cite{as_focs23}. The control of the full (block) error probability still requires  \cite{as_focs23} for RM codes.
\end{abstract}




\begin{definition}
For a  linear code and a symmetric\footnote{A symmetric channel can be viewed as a mixture of binary symmetric channels (BSCs), i.e., independently for each $i$, $Y_i=(\epsilon_i,X_i \oplus w_i)$ with  $\epsilon_i$ in $[0,1/2]$ drawn under some  distribution independently of $w_i \sim \mathrm{Ber}(\epsilon_i)$.} channel, let $X$ be an $n$-dimensional codeword, $Y$ its output on the channel, $\hat{X}(Y)$ the maximum likelihood (ML) decoding of $X$ based off $Y$, $i \in [n]$,  $\hat{X}_i(Y)$ and $\hat{X}_i(Y_{-i})$ the ML decoding of $X_i$ based off $Y$ and $\{Y_j\}_{j \ne i}$ respectively, $P_{\mathrm{glo}
}=P(X\ne \hat{X}(Y))$, $P_{\mathrm{bit},i}=P(X_i\ne \hat{X}_i(Y))$, $P_{\mathrm{bit}}=
\max_{i \in [n]} P_{\mathrm{bit},i}$, $P_{\mathrm{loc},i}=P(X_i\ne \hat{X}_i(Y_{-i}))$, $P_{\mathrm{loc}}=\max_{i \in [n]} P_{\mathrm{loc},i}$.
\end{definition}

Note that the above measures are independent of the codeword choice since the code is linear. Also $P_{\mathrm{bit}}=o_n(1)$ is equivalent to $P_{\mathrm{loc}}=o_n(1)$, and for a BSC$(\epsilon)$, $P_{\mathrm{bit}}<\epsilon \wedge (1-\epsilon)$ implies $P_{\mathrm{loc}}<1/2$.

\begin{lemma} \label{baseError}
For a symmetric channel of capacity $C$ and a linear code of rate $R<C$, there exist $\Omega(n)$ values of $j$ such that $P_{\mathrm{loc},j}=1/2-\Omega(1)$.
\end{lemma}
\begin{proof} 
If the only valid codeword is $0$ then $P_{\mathrm{loc},j}=0$ for all $j$. 
Otherwise let $i$ be such that not all the codewords have their $i$-th bit set to $0$. Since the code is linear, it must be the case that $H(X_i)=1$. Note that $H(Y) \le H(X) + H(Y|X)\le n(H(Y_i)-(C-R))$. Since $H(Y)=\sum_{i=1}^n H(Y_i|Y_{<i})$, the previous inequality implies that there exist $\Omega(n)$ values of $j$ such that $H(Y_j|Y_{<j})=H(Y_i)-\Omega(1)$. For any such $j$ either $H(X_j|Y_{-j})=H(X_j)=0$ or $H(X_j)=1$ and $I(X_j;Y_{-j}) \ge I(Y_j;Y_{-j}) \ge I(Y_j;Y_{<j}) \ge \Omega(1)$; either way, $P_{\mathrm{loc},j}=1/2-\Omega(1)$.
\end{proof}



\begin{definition}
A camellia on $[n]$ centered at $i \in [n]$ with correlation $\rho$ is a collection of subsets of $[n]$ containing $i$ (the petals) such that for any $j\in [n]$ and a uniformly drawn petal $P$, $\mathbb{P}(j \in P )\le \rho$. 
\end{definition}
As discussed in Remark 2, this is quite different than a sunflower.
A camellia code of rate $R_n$,  expansion $\delta_n$ and correlation $\rho_n$ can then be defined as a linear transitive code such that for each $i \in [n]$, (i) there exists a camellia centered at $i$ of correlation $\rho_n$ such that the code's restriction to any petal has rate $\le R_n+\delta_n$, (ii) for each petal in the $i$-centered camellia and each permutation in the symmetry group of the code, if $i$ is contained in the petal's permutation, then the latter is  contained in the  camellia. The following gives a slightly simplified version of this definition.

\begin{definition}
A linear code of rate $R_n$ is a camellia code of expansion $\delta_n$ and correlation $\rho_n$ if it is transitive and has a collection of subsets (its petals) such that the following hold:
(i) the petals are invariant under the symmetry group of the code;
(ii) the code's restriction to any petal has rate $\le R_n+\delta_n$; (iii) given any coordinates $i\ne j$ and a uniformly\footnote{One can further generalize this definition by asking for $\rho_n,s_n$ such that there exists a probability distribution on petals that is preserved by the code symmetry group, such that for any subset $S$ of coordinates with $|S|=s_n\ge 1$ and any $i \notin S$, when drawing a petal $P$ under this distribution $\mathbb{P}(S\subseteq P|i\in P)\le \rho_n$. A uniform measure on the petals with $s_n=1$ recovers the current definition.  
} drawn petal $P$, $\mathbb{P}(j\in P|i\in P)\le \rho_n$.
\end{definition}

\begin{theorem}\label{errorBound}
For a symmetric channel of capacity $C$ and a camellia code of rate $R<C$ with parameters $\delta_n = C-R -\Omega_n(1)$ and $\rho_n=o_n(1)$, we have $P_{\mathrm{loc}}=O(\sqrt{\rho_n})=o_n(1)$. 
\end{theorem}

\begin{definition}
For a petal $P$ of a camellia code, $i \in P$, and a symmetric channel, let $E_{P,i}$ be $1$ (resp.\ $-1$) when the ML decoding of component $i$ given $Y_{P\setminus \{i\}}$ is correct (resp.\ incorrect), defining $E_{P,i}=0$ in case of a tie. 
\end{definition}

\begin{lemma}\label{correlationLem}
For a symmetric channel and a camellia code with correlation $\rho_n$, $i \in [n]$, two petals $P,P'$ uniformly drawn among the petals containing $i$, $\mathbb{E}\mathrm{Cov}(E_{P,i},E_{P',i})\le \sqrt{\rho_n}$.
\end{lemma}

\begin{proof}
Recall that a symmetric channel produces $Y_i=(\epsilon_i,X_i \oplus w_i)$ for each coordinate $i$. Let $Z_i=(\epsilon_i,w_i)$ for each $i$. 
Since the code is linear and the channel is symmetric, $E_{P,i}$ is a function of $Z_{P\setminus \{i\}}$. Let $Q_P(z)$ be the function that returns the value $E_{P,i}$ would take if $Z=z$ (we drop the index $i$ in $Q_P(z)$ as it remains fixed in what follows). For every $S\subseteq P$, define the contribution of $S$ to $Q_P$ as the function such that $Q_{P,S}(z)=\sum_{S'\subseteq S} (-1)^{|S|-|S'|}\mathbb{E}[Q_P(Z)|Z_{S'}=z_{S'}]$. E.g., for the BSC channel, $Q_{P,S}=\langle Q_P, \chi_S \rangle \chi_S$ where $\{\chi_S\}_{S \subseteq [n]}$ is the Fourier basis for $\mathrm{Ber}(\epsilon)^n$. With this definition, $Q_{P,S}(z)$ depends only on $z_S$, $Q_{P,S}$ is orthogonal to any function that is independent of $Z_i$ for any $i\in S$, and $Q_P(z)=\sum_{S\subseteq P\backslash \{i\}} Q_{P,S}(z)$. Thus, given independent random petals $P$ and $P'$ containing $i$, the expected covariance between $E_{P,i}$ and $E_{P',i}$ is

\begin{align}
&\mathbb{E}_{P,P',Z}\left[(Q_P(Z)-\mathbb{E}_{Z'}[Q_P(Z')])\cdot (Q_{P'}(Z)-\mathbb{E}_{Z'}[Q_{P'}(Z')])\right]\\
&=\mathbb{E}_{P,P'} \sum_{S\subseteq P\cap P'\backslash\{i\}: S\ne\emptyset} \mathbb{E}_Z[Q_{P,S}(Z)\cdot Q_{P',S}(Z)] \\
&\le \mathbb{E}_{P,P'} \sum_{S\subseteq P\cap P'\backslash\{i\}: S\ne\emptyset} (\mathbb{E}_Z[Q^2_{P,S}(Z)]\mathbb{E}_Z[Q^2_{P',S}(Z)])^{1/2}  \\
&\le \mathbb{E}_{P,P'} \left(\sum_{S\subseteq P\cap P'\backslash\{i\}: S\ne\emptyset} \mathbb{E}_Z[Q^2_{P,S}(Z)]\sum_{S\subseteq P\cap P'\backslash\{i\}: S\ne\emptyset}\mathbb{E}_Z[Q^2_{P',S}(Z)]\right)^{1/2} \\
&\le \mathbb{E}_{P,P'}\left(\sum_{S\subseteq P\cap P'\backslash\{i\}: S\ne\emptyset} \mathbb{E}_Z[Q^2_{P,S}(Z)]\sum_{S\subseteq P'\backslash\{i\}}\mathbb{E}_Z[Q^2_{P',S}(Z)]\right)^{1/2} \\
&= \mathbb{E}_{P,P'}\left(\left(\sum_{S\subseteq P\cap P'\backslash\{i\}: S\ne\emptyset} \mathbb{E}_Z[Q^2_{P,S}(Z)]\right)\mathbb{E}_Z[Q^2_{P'}(Z)]\right)^{1/2} \\
&\le \mathbb{E}_{P,P'}\left(\sum_{S\subseteq P\cap P'\backslash\{i\}: S\ne\emptyset} \mathbb{E}_Z[Q^2_{P,S}(Z)]\right)^{1/2} .
\end{align}
Further,
\begin{align}
&\mathbb{E}_{P,P'}\left(\sum_{S\subseteq P\cap P'\backslash\{i\}: S\ne\emptyset} \mathbb{E}_Z[Q^2_{P,S}(Z)]\right)^{1/2}
\le \mathbb{E}_{P}\left(\sum_{S\subseteq P\backslash\{i\}: S\ne\emptyset} \mathbb{P}[S\subseteq P']\cdot\mathbb{E}_Z[Q^2_{P,S}(Z)]\right)^{1/2} \\
&\le \mathbb{E}_{P}\left(\sum_{S\subseteq P\backslash\{i\}: S\ne\emptyset} \rho_n\mathbb{E}_Z[Q^2_{P,S}(Z)]\right)^{1/2} \label{rho} 
\le \mathbb{E}_{P}\left(\rho_n \sum_{S\subseteq P\backslash\{i\}} \mathbb{E}_Z[Q^2_{P,S}(Z)]\right)^{1/2} \\
&= \mathbb{E}_{P}(\rho_n \mathbb{E}_Z[Q^2_{P}(Z)])^{1/2} \le \sqrt{\rho_n},
\end{align}
where the first inequality in \eqref{rho} follows from the fact that, knowing that $S\subseteq P\backslash\{i\}, S\ne\emptyset$, there must exists some $j \ne i$ such that $j \in S$, thus $\mathbb{P}[S\subseteq P']\le \mathbb{P}[ j\in P'] \le \rho_n$. 
\end{proof}

\begin{lemma}\label{2moment}
Let $\{E_i\}_{i=1}^k$ be  random variables on $\{-1,0,1\}$ with $\frac{1}{k} \sum_{i\in [k]} \mathbb{E}[E_i]=\Omega(1)$ and $c_k=\frac{1}{k^2}\sum_{i,j \in [k]}\mathrm{Cov}(E_i,E_j)$. Then $\mathrm{Maj}(E_1,\ldots,E_k)=\mathrm{sign}(\sum_{i=1}^k E_i)
)=1$ with probability $1-O(c_k)$.
\end{lemma}

\begin{proof}
$
\mathbb{P}\left(\sum_{i=1}^k E_i\le 0\right) 
\le \mathrm{Var}\left[\sum_{i=1}^k E_i \right]/(\mathbb{E}\left[\sum_{i=1}^k E_i \right])^2=O( k^2c_k)/\Omega(k^2)=O(c_k)$.
\end{proof}

\begin{proof}[Proof of Theorem 1]
First, observe that by Lemma \ref{baseError}, there exists $c>0$ such that for any petal $P$ and uniformly drawn $i$ in $P$, there is a nonvanishing probability that we can recover the value of $X_i$ from $Y_{P\backslash\{i\}}$ with accuracy at least $1/2+c$. Also, for other choices of $i$ we can still recover $X_i$ from $Y_{P\backslash\{i\}}$ with accuracy at least $1/2$. Thus, $\mathbb{E}_{i,Y}(E_{P,i})=\Omega(1)$. That in turn means that if we pick a pair of a vertex $i$ and a petal $P$ such that $P$ contains $i$, uniformly at random, $\mathbb{E}_{P, i,Y}(E_{P,i})=\Omega(1)$. By transitivity of the code and invariance of the collection of petals over the code's symmetry group, this is independent of the value of $i$; so for any $i$, it is the case that for a uniformly drawn $P$  containing it,  $\mathbb{E}_{P,Y}(E_{P,i})=\Omega(1)$. Also, given two petals $P$ and $P'$ containing $i$ drawn independently and uniformly at random,
the expected covariance between them is at most $\sqrt{\rho_n}$ by Lemma \ref{correlationLem}. 
The conclusion then follows from Lemma \ref{2moment}.
\end{proof}

The Reed-Muller code $RM(m,r)$ is the code obtained by evaluating polynomials of degree at most $r$ on the $m$-dimensional Boolean hypercube (thus $n=2^m$ and $R_n=\sum_{i=0}^r {m \choose i}/n$). 
\begin{lemma}
An $RM$ code with $r<m-1$ is a camellia code of parameters $\delta_n=o_n(1),\rho_n=O(2^{-2\sqrt{\log(n)}/\log\log(n)})$. 
\end{lemma}

\begin{proof}
We claim that the collection of cosets of subspaces of dimension $\log(n)-2\sqrt{\log(n)}/\log\log(n)$ satisfies the conditions to be the code's petals for sufficiently large $n$. This collection is invariant under the symmetry group of the RM code due to the affine-invariance of the code and the fact that the symmetry group is exactly the group of affine transformations on $\mathbb{F}_2^m$ for $r<m-1$. The restriction of the code to any of these subspace cosets has a rate $o_n(1)$ greater than that of the overall code since ${m \choose i}=O(n/\sqrt{m})$ for all $i$. Also, any such subspace coset contains $n/2^{2\sqrt{\log(n)}/\log\log(n)}$ vertices. Whether or not a random coset in this collection contains a given vertex is independent of whether or not it contains any other vertex; for any coordinates $i\ne j$, the probability that a random one of these cosets contains $j$ conditioned on it containing $i$ is thus $O(2^{-2\sqrt{\log(n)}/\log\log(n)})$. 
\end{proof}
\begin{corollary}\label{RMbit}
A Reed-Muller code of asymptotic rate below capacity on any symmetric channel has bit-error probability vanishing as $P_{\mathrm{bit}}=O(2^{-\sqrt{\log(n)}/\log\log(n)})$.
\end{corollary}

{\bf Remarks.} 
(1) Corollary \ref{RMbit} gives an alternative proof to the bit-error result derived in \cite{Kudekar17}, without relying on the area theorem and with a self-contained proof that does not rely on sharp thresholds for monotone Boolean properties \cite{friedgut1996every}. Further, Corollary \ref{RMbit} applies to any symmetric channel, and gives a scaling of $P_{\mathrm{loc}}=O(\sqrt{\rho_n
}
)=O(2^{-\sqrt{\log(n)}/\log\log(n)})$, compared to the scaling of $O(\log\log(n)/\sqrt{\log(n)})$ in \cite{reeves}, with a simpler proof. 
(2) Camellia codes are codes affording a structure amenable for boosting procedures, i.e., decoding local components and merging them. The components need to have controlled dependencies, granted by the camellia structure. Compared to \cite{as_focs23},  camellia boosting gives a more general framework than sunflower boosting: in a sunflower, every petal is required to contain the kernel and they are not allowed to intersect anywhere else. The requirement on petals in a camellia is almost the opposite: they are forbidden from having any pair of vertices that they are excessively likely to contain but otherwise free to intersect anywhere (since a fixed pair of vertices should be contained in a limited number of petals, the petals overlap cannot be too large, but the overlap can be spread).  
In particular, in a doubly transitive code, taking the orbit of the petals in a sunflower under the code's symmetry group will yield a camellia (if the code's restriction to the petals has appropriate rate), while there is no obvious way to convert a general camellia to a sunflower. Also, we can go directly from knowing that a code is a camellia code to knowing that it has vanishing local error, while with a sunflower one needs to prove bounds on the probability that the restriction of the error vector to the kernel is ``bad" (in the sense that its extensions to the petals are likely to result in incorrect decodings of the target bit).  (3) The proof technique extends directly to $q$-ary symmetric input channels and $q$-ary RM codes when $q$ is prime. (4) Note that Theorem 1 only bounds the bit or local error probability;  the control of the block or global error probability is covered in \cite{as_focs23} for RM codes with significantly more work. Note that the `recursive boosting' steps of \cite{as_focs23} may adapt to codes with recursive camellia structures.  (5) An asset of the boosting framework is that it requires a weak base-case to get started, which can be obtained in great generality via the `entropy bending' argument (Lemma 1). Further, given an algorithmic black-box for the base-case, the camellia boosting algorithm is efficient if the camellia can be constructed efficiently, which is the case for RM codes. This thus gives an algorithmic reduction: if one can efficiently obtain the base-case (i.e., decode a bit in the RM code given the other noisy bits with error probability $1/2-2^{-c\sqrt{m}}$ for a sufficiently small constant $c$), then the camellia boosting algorithm is overall efficient.


\newcommand{\etalchar}[1]{$^{#1}$}

\end{document}